\documentclass[11pt]{amsart}
\usepackage{graphicx}
\usepackage{hyperref}
\usepackage{tensor}
\usepackage{xcolor}
\usepackage{thmtools}
\usepackage{amsaddr} 
\usepackage{tikz} 
\usetikzlibrary{matrix}

\makeatletter
\let\uppercasenonmath\@gobble
\makeatother

\title{K\"ahler Representation Theory}
\author{Bryan W. Roberts} 
\address{Philosophy, Logic \& Scientific Method\\Centre for Philosophy of Natural and Social Sciences\\London School of Economics \& Political Science\\Houghton Street, London WC2A 2AE, UK}
\email{\href{mailto:b.w.roberts@lse.ac.uk}{b.w.roberts@lse.ac.uk}}

\author{Nicholas J. Teh}
\address{Department of Philosophy\\University of Notre Dame\\100 Malloy Hall\\Notre Dame, IN 46556, USA} 
\email{\href{mailto:nteh@nd.edu}{nteh@nd.edu}}


\newcommand{\Al}{\mathfrak{A}} 
\newcommand{\Alg}{\mathcal{A}} 
\newcommand{\PP}{\mathcal{P}} 
\newcommand{\CC}{\mathbb{C}} 
\newcommand{\RR}{\mathbb{R}} 
\newcommand{\HH}{\mathcal{H}} 

\newcommand{\br}[1]{\{#1\}}
\newcommand{\Cat}[1]{\mathbf{Cat}(#1)}

\newtheorem{lemm}{Lemma}
\newtheorem{jlb}{JLB}
\newtheorem*{thm}{Theorem}
\theoremstyle{definition}
\newtheorem*{defi}{Definition}
 
\begin{document}
\begin{abstract}
We show that Jordan-Lie-Banach algebras, which provide an abstract characterization of quantum theory equivalent to $C^*$ algebras, can always be canonically represented in terms of smooth functions on a K\"ahler manifold.
\end{abstract}
\maketitle

\section{Introduction}

An alternative to the usual algebraic formulation of quantum theory is the geometric formulation on K\"ahler manifolds, or ``K\"ahler quantum theory''. In rough sketch, the idea is to treat the complex projective Hilbert space $\PP\HH$ as a real manifold with a complex structure. The geometry then derives from the Hilbert space inner product: the real part of the inner product gives rise to a Riemannian (K\"ahler) metric with constant holomorphic sectional curvature, while the pure imaginary part gives rise to a symplectic form.

K\"ahler quantum theory provides one convenient way to compare quantum theory to analytic mechanics, and to other theories with natural geometric descriptions. It also provides an interesting new perspective on some of the central ideas of quantum theory, including the uncertainty relations, Schr\"odinger evolution, and entanglement, all of which can be given informative geometric expressions. The basic features of this formulation were first identified in 1979 by Kibble \cite{kibble1979geometrization}. However, a great deal is now known about K\"ahler quantuum theory thanks to more recent work \cite{CirelliEtAl1990, Gibbons1992147, schilling-dissertation, AshtekSchill1999a, landsman1997transitionprob, brodyhughston1998, brody2001geometric}. Some philosophical and foundational consequences have also been explored in \cite{brody2001geometric,teh2012cloning,roberts-2014time}.

In this paper we would like to offer a new perspective on K\"ahler quantum theory, and on geometric quantum theories more generally, by viewing them as the products of representation theory. The general technique that we propose lies in the tradition of identifying an appropriate algebra, and understanding how that algebra and its operations can be naturally represented by objects on a manifold. The particular algebra that we adopt here is a class of Jordan algebra sometimes called a Jordan-Lie-Banach (JLB) algebra. Our main result is then to show that a JLB algebra can always be canonically represented as a natural algebra of functions on a K\"ahler manifold. This is directly analogous to classic Gelfand-Naimark-Segal (GNS) theorem \cite{GelfandNaimark1943gns,segal1947gns}, which found that a state over a $C^*$-algebra can be canonically represented as an algebra of bounded linear Hilbert space operators.

This connection to JLB algebras bears on the foundational question of how to define ``observables''. In orthodox quantum theory, self-adjoint operators play a unique interpretive role as the representatives of observable physical quantities. But such operators do not play any correspondingly unique role in the structure of a $C^*$ algebra, which contains many non-self-adjoint operators in addition to the self-adjoint ones. In particular, the matrix product $ab$ of two self-adjoint operators fails to be self-adjoint whenever $a$ and $b$ do not commute. This led Jordan \cite{jordan1933a} to propose that operator products in quantum theory be characterised by symmetrisation, or what is sometimes called a regular Jordan product, which is always self-adjoint:
\[
  a\circ b := \tfrac{1}{2}({ab + ba}).
\]
Jordan, von Neumann and Wigner \cite{JordvNeuWign1934,vonNeumann1936a} then tried to characterise the structure of quantum theory in terms of a Banach space equipped with a Jordan product, known as a Jordan-Banach (JB) algebra. However, it was soon realised there exist JB algebras that are not isomorphic to the self-adjoint part of a $C^*$ algebra, which suggested that JB algebras are inadequate for the description of quantum theory. Progress came when Connes \cite{connes1974a} discovered the conditions under which such an isomorphism exists. These turned out to be equivalent to the addition of a derivation structure \cite{AlfsenShultz1998a}, which Alfsen and Schultz define in terms of an order derivation \cite[Theorem 6.15]{AlfsenSchultz2003book}. We shall follow Landsman \cite{landsman1998mathematical} in defining it in terms of the Lie bracket of a JLB algebra. 

The observation of this paper is that, whereas Hilbert spaces provide a natural way to represent $C^*$ algebras, K\"ahler manifolds provide a natural way to represent JLB algebras. The relation between these two approaches is illustrated in Figure \ref{fig:commuting}. Along the top edge of the commuting diagram, a $C^*$ algebra is mapped to a JLB algebra by a particular standard of equivalence, which we shall observe is in fact a forgetful functor; this relationship has been expressed in similar ways by others \cite{AlfsenHanche-OlsenSchultz1980a,shultz1982a,brown1992a,landsman1997transitionprob}. The bottom arrow expresses the construction of K\"ahler quantum mechanics from a Hilbert space representation. The left-edge describes the ordinary Hilbert space representation theorem established by the GNS construction. The aim of this paper is to construct the right-edge, by proposing a representation $\pi$ from a JLB algebra to a particular algebra of functions on a K\"ahler manifold $K=(M,g,\Omega,J)$, where $M$ is a smooth manifold, $g$ a Riemannian metric, $J$ a complex structure and $\Omega$ a (symplectic) K\"ahler form.

\begin{figure}[tbh]\begin{center}
    \begin{tikzpicture}
  \matrix (m) [matrix of math nodes,row sep=3em,column sep=4em,minimum width=2em] {
     \Al_{C^*} & \Alg_{JLB} \\
     (\HH,\pi) & (K,\pi) \\};
   \path
       (m-1-1) edge [-stealth] node [left] {GNS} (m-2-1)
       (m-1-1) edge [-stealth] node [above] {$\cong$} (m-1-2)
       (m-2-1) edge [-stealth] node [below] {} (m-2-2)
       (m-1-2) edge [-stealth] node [right] {Rep} (m-2-2);
\end{tikzpicture}
    \caption{K\"ahler quantum mechanics as a representation theory for JLB algebras}\label{fig:commuting}
  \end{center}\end{figure}

\section{Elementary Definitions}

A \emph{$C^*$ algebra} $\Al_{C^*}$ is a complex linear Banach space, together with an associative product $ab$ and an involution ${}^*$ that satisfies,
\begin{align*}
  \text{(Cauchy-Schwarz)} \hspace{2pc} & |ab| \leq |a||b|,\\
  \text{($C^*$-property)}  \hspace{2pc} & |a^*a|=a^2.
\end{align*}

A \emph{Jordan-Lie-Banach (JLB) algebra} $\Alg_{JLB}$ is a real linear Banach space, together with a bilinear commutative product $a\circ b=b\circ a$ and an antisymmetric bracket $\br{a,b} = -\br{b,a}$ that satisfies:
\begin{align*}
  \text{(Leibniz rule)} \hspace{2pc} & \br{a\circ b,c} = a\circ\br{b,c} + b\circ\br{a,c},\\
  \text{(associator identity)} \hspace{2pc} & (a\circ b) \circ c - a \circ (b \circ c) = \{\{a,c\},b\},\\
  \text{(Jordan identity)} \hspace{2pc} & \br{\br{a,b},c} + \br{\br{b,c},a} + \br{\br{c,a},b}=0\\
  \text{(Cauchy-Schwarz)} \hspace{2pc} & |a\circ b|\leq|a||b|,\\
  \text{(triangle inequality)} \hspace{2pc} & |a|^2 \leq |a\circ a + b\circ b|.
\end{align*}
We also presume a JLB algebra contains a unit  $\textbf{1}\in\Alg_{JLB}$ satisfying $\mathbf{1}\circ a=a$ for all $a$; if a JLB algebra does not contain a unit, then one can always be added. These axioms imply the Jordan identity, and so a JLB algebra is at once a unital Jordan algebra, a Lie algebra, and a Banach space\footnote{See \cite{landsman1998mathematical} for a clear development.}.

A \emph{state on a $C^*$ algebra} is a complex linear functional $\varphi:\Al_{C^*}\rightarrow\CC$, which is both positive in the sense that $\varphi(a^*a)\geq 0$ and also normalised, $|\varphi|=1$, where $|\varphi| := \sup\br{\varphi(a) \,|\, |a|\leq 1}$. A \emph{state on a JLB algebra} is a real linear functional $\varphi:\Al_{JLB}\rightarrow\RR$, which is positive in the sense that $\varphi(a\circ a) \geq 0$, and also normalised, $\varphi(\mathbf{1})=1$.

A \emph{K\"ahler manifold} is a quadruple $K=(M,g,\Omega,J)$, where $M$ is a smooth manifold, $g$ is a Riemannian metric, $\Omega$ is a symplectic form, and $J$ is a complex structure, and where these objects satisfy,
\begin{align*}
  \text{(compatibility)} \hspace{2pc} & g(JX,JY)=g(X,Y), \;\;\; \Omega(JX,JY)=\Omega(X,Y)\\
  \text{(K\"ahler property)} \hspace{2pc} & \Omega(X,Y) = g(JX,Y)
  \end{align*}
for all vectors $X,Y$ in the tangent bundle. Note that it follows from these definitions that $\Omega(X,JY) = g(X,Y)$ and $\Omega(JX,Y) = -g(X,Y)$, and also that $g(X,JY)=-\Omega(X,Y)$.

\section{The C*-JLB Correspondence}\label{sec:c*-jlb-correspondence}

The basic correspondence between $C^*$ algebras and JLB algebras is analogous to the correspondence between the complex and real fields. First, for each $C^*$-algebra $\Al_{C^*}$, there exists a unique $JLB$-algebra $\Alg_{JLB}$ consisting of the restriction of $\Al_{C^*}$ to its self-adjoint elements. The norm is inherited from $\Al_{C^*}$, and the Jordan and Lie products are defined by,
\begin{align*}
	& a\circ b := \tfrac{1}{2}(ab+ba),\\
    & \{a,b\} := \tfrac{1}{2i}(ab - ba).
\end{align*}
Conversely, for each JLB-algebra $\Alg_{JLB}$, there exists a unique $C^*$-algebra $\Al_{C^*}$ consisting of the closure of $\Al_{JLB}$ under scalar multiplication by the complex numbers. The product, involution and norm are defined by,
\begin{align*}
	& ab := a\circ b + i\{a,b\}\\
	& (a + ib)^* := a - ib\\
	& |a| := |a^*a|^{1/2}.
\end{align*}
Note that for every element $c$ of a $C^*$ algebra there is a unique pair of self-adjoint operators $a,b$ such that $c=a+bi$.

The character of this relationship is a little clearer from the categorical perspective. Let $\Cat{C^*}$ be the category of $C^*$ algebras, in which the objects are the $C^*$ algebras and the morphisms are $*$-homomorphisms. Let $\Cat{JLB}$ be the category of JLB algebras, in which the objects are JLB algebras and the morhpisms are JLB-homomorphisms, i.e. the continuous functions that preserve the Jordan and Lie products.

These categories are known to be uniquely recoverable from each other via the prescription above \cite{landsman1998mathematical,HalvorsonClifton1999a}. However, there is also a sense in which they are not truly equivalent, because some structure is added in the construction of a $C^*$-algebras from a JLB algebra, again analogous to the complexification of the real number field. This can be made precise by observing that the construction of a $C^*$ algebra from a JLB algebra above is characterised by a functor $F:\Cat{JLB}\rightarrow\Cat{C^*}$ that is faithful and essentially surjective, but not full.

This can be seen as follows. For each JLB algebra $\Alg\in\Cat{JLB}$ let $F(\Alg)$ be the $C^*$ algebra defined by the prescription above; and for each JLB homomorphism $f:\Alg\rightarrow\Alg^\prime$ let $F(f)$ be the $*$-homomorphism defined by $F(f)(a+bi) = f(a)+f(b)i$. Then $F$ is easily seen to be a functor, since it preserves the identity element and $F(\varphi_1\circ\varphi_2)=F(\varphi_1)\circ F(\varphi_2)$. It is essentially surjective, since every $C^*$ algebra is the image of some $JLB$ algebra under $F$. It is faithful because for each pair of JLB homomorphisms $f,g$, we have that $F(f)=F(g)$ only if $f(a) + if(b) = g(a) + ig(b)$ which implies that $f=g$. However, it is not full, since involution ${}^*$ generates morphisms of $\Cat{C^*}$ that have no analogue in $\Cat{JLB}$.

A complete equivalence of categories can be restored by considering ordered pairs (the complexification of $\Cat{JLB}$) instead of single JLB algebras, and one should expect pairs of JLB algebras to arise naturally in the JLB description of quantum theory. We shall thus adopt pairs of JLB elements as our basic objects of study in the construction below.

\section{K\"ahler representation theory}\label{sec:geometric-representation-theory}

The basic technique of this paper is to represent JLB algebras in a way that is analogous to the GNS representation of $C^*$ algebras. The GNS construction shows how a $C^*$ algbera and a state $(\Al_{C^*},\varphi)$ can be associated with a canonical triple $(\HH,\pi(\Al_{C^*}),\psi)$, where $\HH$ is a Hilbert space, $\pi(\Al_{C^*})$ is a representation of $\Al_{C^*}$ in terms of bounded operators on $\HH$, and $\psi$ is a preferred vector. K\"ahler quantum theory can be viewed as a similarly direct representation of a JLB algebra. In particular, we shall see that a JLB algebra and a state $(\Alg_{JLB},\varphi)$ are associated with a canonical triple $(K,\pi,\nu)$, where $K=(M,g,\Omega,J)$ is a K\"ahler manifold, $\pi(\Alg_{JLB})$ is a representation of $\Alg_{JLB}$ into the algebra of smooth functions on $M$, and $\nu$ is a preferred point in $M$.

Let us begin by making our notion of a representation precise. The metric $g$ and the symplectic form $\Omega$ of a K\"ahler manifold each determine a natural algebra on the set of smooth functions $C^\infty(M)$. For each smooth function $f:P\rightarrow\RR$, let $X_f$ be the be the corresponding Hamiltonian vector field satisfying $\iota_{X_f}\Omega = df$. Then $g$ and $\Omega$ each induce a binary operation on  $C^\infty(M)$ given by,
\begin{align}\label{eq:induced-algebra}
\begin{split}
  & f \circ g := \tfrac{1}{2}g(X_f,X_g)\\
  & \{f,g\} := \tfrac{1}{2}\Omega(X_f,X_g).
\end{split}
\end{align}
This defines an algebra of smooth functions on any K\"ahler manifold, which may be referred to as the \emph{algebra of smooth functions on a K\"ahler manifold $K$}, and which immediately suggests the following definition.

\begin{defi}
Let $\Alg_{JLB}$ be a JLB algebra and let $K=(M,\Omega,g,J)$ be a K\"ahler manifold. If $\pi:a\mapsto f_a$ is a mapping from $\Alg$ into the algebra of smooth functions on a K\"ahler manifold defined by Equation \eqref{eq:induced-algebra}, and if this mapping satisfies,
\begin{align*}
  f_a\circ f_b & = f_{a\circ b}\\
  \{f_a, f_b\} & = f_{\{a,b\}}.\\
  |f_a| & \leq |a|
\end{align*}
for all $a,b\in\Alg_{JLB}$, then the pair $(K,\pi)$ is called a \emph{K\"ahler representation} of $\Alg_{JLB}$.
\end{defi}

\section{The canonical K\"ahler representation}

The main result of this paper is the following.

\begin{thm}
Every state over a JLB algebra $(\Alg_{JLB},\varphi)$ admits a K\"ahler representation $K = (M,\Omega,g,J)$, $\pi:a\mapsto f_a$ and a preferred point point $\nu\in M$ such that the triple $(K,\pi,\nu)$ is the unique one to satisfy the following conditions:
\begin{enumerate}
\item $f_a(\nu) = \varphi(a)$ for all $a\in\Alg$; and
\item (K\"ahler Cyclicity) Any two points in $M$ are connected by a piecewise smooth Hamiltonian curve generated by the Hamiltonian vector fields associated with $\{ f_a \}$,
\end{enumerate}
where uniqueness is defined up to K\"ahler isomorphism, i.e. a diffeomorphism that preserves both $g$ and $\Omega$.
\end{thm}
Our proof comes in three stages: construction of the K\"ahler manifold $K$, construction of the representation $\pi$, and establishing uniqueness. We proceed with each in turn.

\subsection{Construction of the K\"ahler manifold}\label{subs:construction-of-the-kahler-manifold}

Motivated by the discussion of Section \ref{sec:c*-jlb-correspondence}, our rough strategy is to take $M$ as the set of pairs $(a,b)\in\Alg\times\Alg$. We say ``rough'' because, as in the GNS construction, one must first quotient out by the kernel of $\varphi$ in order to avoid degeneracies in the construction.

The Cartesian product $\Alg\times\Alg$ is a linear vector space over $\RR\times\RR$ with $(a,b)+(a^\prime,b^\prime):=(a+a^\prime,b+b^\prime)$. Therefore it is also a linear manifold. We define the \emph{kernel} of $\varphi:\Alg\times\Alg\rightarrow\RR$ of this manifold to be,
\[
  \ker\varphi := \{(a,b) \,|\, \varphi(a\circ a)=\varphi(b\circ b)=\varphi(\br{a,b})=0.
\]
The kernel is a linear subspace of $\Alg\times\Alg$. We can thus write $\xi_{(a,b)}$ to denote an equivalence class of the form,
\[
  \xi_{(a,b)} :=\{(a^\prime,b^\prime)\in\Alg \,|\, (a-a^\prime,b-b^\prime) \in \ker \varphi\}, 
\]
and write the resulting quotient space as $M = \Alg\times\Alg/\ker\varphi$. This is the desired linear manifold of our construction.

We define a complex structure on $M$ by the linear mapping,
\[
  J\xi_{(a,b)} := \xi_{(-b,a)}.
\]
This mapping is well-defined, in that it returns the same values for different members of the same equivalence class. It also allows us to adopt a useful shorthand that we will use throughout the rest of the argument. Namely, when it is clear that we are speaking about an element of $M := \Alg\times\Alg/\ker\varphi$, let us write,
\[
  \xi_a := \xi_{(a,0)}.
\]
Since $J\xi_b = \xi_{(0,b)}$, we can write an arbitrary point in $M$ as,
\[
  \xi_{(a,b)} = \xi_a + J\xi_b.
\]
This provides a neat notation with which to define our metric $g$ and symplectic form $\Omega$. We begin by defining two bilinear functionals on $M$, but considering firstly the case of elements of $M$ of the form $\xi_{(a,0)} = \xi_a$,
\begin{align*}
  g(\xi_a,\xi_b) := \varphi(a\circ b) &&    \Omega(\xi_a,\xi_b) := \varphi(\br{a,b}).
\end{align*}
We now extend these definitions to elements of the form $\xi_{(0,a)} = J\xi_a$ by defining,
\begin{align*}
  g(J\xi_a,\xi_b) :=  \Omega(\xi_a,\xi_b) &&    \Omega(\xi_a,J\xi_b) := g(\xi_a,\xi_b)\\
  g(\xi_a,J\xi_b) :=  -\Omega(\xi_a,\xi_b) &&    \Omega(J\xi_a,\xi_b) := -g(\xi_a,\xi_b).
\end{align*}
This is enough to define $g$ and $\Omega$ on all of $M$, since both are taken to be linear. For example, linearity implies $g(\xi_a+J\xi_b,\xi_c+J\xi_d) = g(\xi_a,\xi_c) + g(J\xi_b,\xi_c) + g(\xi_a,J\xi_d) + g(J\xi_b,J\xi_d)$. The first three terms are immediately given by our definitions, and the fourth term is equal to $\Omega(\xi_b,J\xi_d)$, which is again given by our definitions. Note that these definitions imply $\Omega(J\xi_a,J\xi_b) = \Omega(\xi_a,\xi_b)$ and $g(J\xi_a,J\xi_b) = g(\xi_a,\xi_b)$, so both compatibility and the K\"ahler property are satisfied by these functionals.

One can further verify that $\Omega$ and $g$ are well-defined in the sense that they take the same value when we replace $a\in\xi_a$ with another element $a^\prime$ of the same equivalence class. In particular, if $a,a^\prime\in\xi_a$, then $(a-a^\prime,b)\in\ker\varphi$, and thus $\varphi((a-a^\prime)\circ(a-a^\prime))=0$. But since states on a JLB algebra obey the Cauchy-Schwarz inequalities (Appendix JLB-\ref{jlb:1}), $\varphi(c\circ c)=0$ only if $\varphi(c\circ b) = \varphi(\br{c,b}) = 0$ for all $b\in\Alg$. Therefore, letting $c=a-a^\prime$, we have,
\begin{align*}
  g(\xi_a,\xi_b) - g(\xi_{a^\prime},\xi_b) & = \varphi((a - a^\prime)\circ b) = 0, \text{ and}\\
  \Omega(\xi_a,\xi_b) - \Omega(\xi_{a^\prime},\xi_b) & = \varphi(\br{a - a^\prime, b}) = 0
\end{align*}
for all $b\in\Alg$. Thus, $ g(\xi_a,\xi_b) = g(\xi_{a^\prime},\xi_b)$ and $\Omega(\xi_a,\xi_b) = \Omega(\xi_{a^\prime},\xi_b)$. Symmetric reasoning holds of the right-hand argument, and for elements of the form $J\xi_a$, and so $g$ and $\Omega$ are well-defined as functionals on $M$.

It is now a straightforward exercise to see that $g$ is symmetric, positive and non-degenerate. It is symmetric because the Jordan product $\circ$ in our JLB algebra is commutative. To see that it is positive, observe that by our definitions,
\begin{equation}\label{eq:expanded-g}
  g(\xi_a+J\xi_b,\xi_c+J\xi_d) = \varphi(a\circ c) - \varphi(\br{a,d}) + \varphi(\br{b,c}) + \varphi(b\circ d).
\end{equation}
Let $c=a$ and $d=b$. Then we find that $g(\xi_a+J\xi_b,\xi_a+J\xi_b) = \varphi(a\circ a) + \varphi(b\circ b)$, and hence that $g$ is positive by the positivity of $\varphi$. To see that it is non-degenerate, suppose that $g(\xi_a+J\xi_b,\xi_c+J\xi_d)=0$ for all $\xi_c + J\xi_d$. Then from the same case, $c=a$ and $d=b$, it follows that $\varphi(a\circ a) = \varphi(b\circ b) = 0$. And from the case that $c=d=a$ we have that $\varphi(\br{a,b}) = \varphi(a\circ a) = 0$. Hence, $\xi_{(a,b)}= \xi_{(0,0)}$ is the zero element, and $g$ is non-degenerate.

It is also straightforward to show $\Omega$ is antisymmetric and non-degenerate. It is antisymmetric because $\br{,}$ is too. But we also know that $\Omega(\xi_a+J\xi_b,\xi_c+J\xi_d) = g(\xi_a+J\xi_b,\xi_d-J\xi_c)$. So, since $g$ is non-degenerate, it follows that $\Omega$ is non-degenerate as well.

Our final step is to lift these functionals to tensor fields $g$ and $\Omega$ on $M$; we abuse notation slightly by using the same symbols $g,\Omega$ for both the functionals on $M$ and the tensor fields; it should be clear from context which is which. Since $M$ is a linear manifold, each point $\xi$ in the manifold defines a vector $Z_\xi$ in the tangent space of every other point $\xi^\prime$, given by $Z_\xi(f) := \tfrac{d}{ds}f(\xi^\prime + s\xi)|_{s=0}$. The functionals $g$ and $\Omega$ thus lift to tensor fields by the definition that in the tangent space at every point, 
\begin{align*}
  g(Z_{\xi_1}, Z_{\xi_2}) & := g(\xi_1,\xi_2)\\
  \Omega(Z_{\xi_1}, Z_{\xi_2}) & := \Omega(\xi_1,\xi_2),
\end{align*}
and similarly for the 1-1 tensor $J$. From the discussion above it follows that $g$ is a Riemannian metric, $\Omega$ is an almost symplectic form and $J$ is a complex structure.

We finally apply the elementary result that an almost symplectic manifold $(M,\Omega)$ is symplectic if and only if the Poisson bracket $\{,\}$ associated with $\Omega$ satisfies the Jacobi identity\footnote{For example, see \cite{wurzbacher2001a} pgs. 110-111, or \cite{landsman1998mathematical} Proposition 2.3.7.}. If three functions have the form $f_a$, $g_b$, $h_c$, where by definition $f_a(\xi_b) = g(\xi_b,\xi_{ab})$, then the Jacobi identity is satisfied. For by the definitions of these functions it is straightforward to show that $\br{f_a,\br{g_b,h_c}} + \br{g_b,\br{h_c,f_a}} + \br{h_c,\br{f_a,g_b}}$ is equal to,
\[
 \varphi\left(\br{a,\br{b,c}} + \br{b,\br{c,a}} + \br{c,\br{a,b}}\right), 
\]
which vanishes because of the Jacobi axiom for JLB algebras. Now fix any point $\xi$ in the manifold, and consider the value of an arbitrary function $f$ there. We can find always find a JLB element $a$ such that $f_a(\xi_b)=f(\xi_b)$, for example by scaling $a$ by a constant until an appropriate value is reached. So, by replacing three arbitrary functions $f,g,h$ with functions of the form $f_a$, $g_b$, $h_c$ at $\xi_b$, we find that the Jacobi identify is satisfied by general smooth functions, and so $(M,\Omega)$ is symplectic.

Combining all the above observations, it is evident that we can construct a K\"ahler manifold $K=(M,g,\Omega, J)$ from a state over a JLB algebra. We shall now see how this manifold encodes the JLB algebra.

\subsection{Construction of the representation}

In Section \ref{sec:geometric-representation-theory} we observed that the smooth functions on an arbitrary K\"ahler manifold admit two natural binary operations, $f\circ g := \frac{1}{2}g(X_f,X_g)$ and $\br{f,g} := \tfrac{1}{2}\Omega(X_f,X_g)$, which we used to define the notion of a K\"ahler representation. We now show that, on the K\"ahler manifold constructed by the prescription above, the resulting algebra of functions is a K\"ahler representation.

We will define a representation $\pi:a\mapsto f_a$, which associates each element $a\in\Alg$ of the JLB algebra with a smooth function $f_a\in C^\infty(M)$. Write an arbitrary point in the manifold as $\xi_{(b,c)} = \xi_b + J\xi_c$. Then we take $f_a$ to be the function defined by,

\begin{align}\label{eq:f_a}
\begin{split}
    f_a(\xi_b) & := \tfrac{1}{2}g(\xi_b,\xi_{a\circ b}) - \tfrac{1}{2}\Omega(\xi_b,\xi_{\br{a,b}})\\
    f_a(J\xi_c) & := \tfrac{1}{2}g(\xi_c,\xi_{a\circ c}) - \tfrac{1}{2}\Omega(\xi_c,\xi_{\br{a,c}})\\
    f_a(\xi_b+J\xi_c) & :=  f_a(\xi_b) +  f_a(J\xi_c).
\end{split}
\end{align}
It is perhaps more illuminating to think about this definition in terms of a special point $\xi_{ab}\in M$ associated with two elements $a,b\in \Alg$, defined by,
\begin{equation}\label{eq:ab-def}
  \xi_{ab} := \xi_{a\circ b} + J\xi_{\br{a,b}}.
\end{equation}

Then since $-\Omega(\xi_b,\xi_{\br{a,b}}) = g(\xi_b,J\xi_{\br{a,b}})$, we can equivalently write our representation mapping in Equation \ref{eq:f_a} as,
\begin{align}\label{eq:f_a2}
  \begin{split}
  & f_a(\xi_b) = f_a(J\xi_b) = \tfrac{1}{2}g(\xi_b,\xi_{ab}),\\
  & f_a(\xi_b+J\xi_c) =  f_a(\xi_b) +  f_a(J\xi_c).
  \end{split}
\end{align} 
Those familiar with the literature on geometric quantum mechanics will find this suggestive. If we view $\circ$ and $\br{,}$ as being associated with a decomposition into real and complex parts, then this $f_a$ is exactly what turns out to express the expectation value of a quantum observable (see for example Equation 2.4 of \cite{AshtekSchill1999a}). Here we are going the other direction: instead of beginning with a geometric quantum system and deriving Equation \eqref{eq:f_a2}, we adopt Equation \eqref{eq:f_a2} as a definition, and then derive a geometric quantum system using the structure of the JLB algebra.

We begin with a lemma, which makes use of the definition of $\xi_{ab}$ in Equation \ref{eq:ab-def} above.

\begin{lemm}\label{lemm:sa}
  $g(\xi_a,\xi_{bc}) = g(\xi_{ba},\xi_c)$.
\end{lemm}
\begin{proof}
We begin by writing,
\[
  g(\xi_a,\xi_{bc}) = g(\xi_a,\xi_{b\circ c}) - \Omega(\xi_a, \xi_{\br{b,c}}) = \varphi(a\circ (b\circ c) - \br{a,\br{b,c}})
\]
Applying the associator identity to the first term and the Jacobi identity to the second term we get,
  \begin{align*}
    g(\xi_a,\xi_{bc})
    & = \varphi(\underbrace{(a\circ b)\circ c - \br{\br{a,c},b}}_{\text{Assoc. id}} + \underbrace{\br{b,\br{c,a}} + \br{c,\br{a,b}}}_{\text{Jacobi id}})\\
    & = \varphi((a\circ b)\circ c + \br{c,\br{a,b}}) = g(\xi_{b\circ a},\xi_c) + \Omega(\xi_{\br{b,a}},\xi_c)\\
    & = g(\xi_{b\circ a},\xi_c) + g(J\xi_{\br{b,a}},\xi_c) = g(\xi_{ba},\xi_c).
  \end{align*}
\end{proof}

The next lemma shows that our definition of $\pi:a\mapsto f_a$ satisfies a central property in geometric quantum mechanics, that the Hamiltonian vector field $X_b$ it generates is equivalent to what has be called its ``Schr\"odinger'' vector field \cite{schilling-dissertation,AshtekSchill1999a}.

Let $X_b$ denote the Hamiltonian vector field generated by $f_b$, i.e. the unique vector field satisfying $\iota_{X_b}\Omega = df_{b}$. Then the \emph{Schr\"odinger vector field} $Y_b$ generated by $f_b$ at a point $\xi_a\in M$ is given by,
\[
  Y_b := -J(X_{ba}). 
\]
Then we have the following, which again makes use of our definition of $\xi_{ab}$ from Equation \ref{eq:ab-def}.

\begin{lemm}[Schr\"odinger-Hamiltonian Equivalence]\label{lemm:schrodinger} 
  Let $f_b(\xi_c) = \tfrac{1}{2}g(\xi_c,\xi_{bc})$. If $Y_b$ is the Schr\"odinger vector field generated by $f_b$ and $X_b$ is the Hamiltonian vector field, then $Y_b = X_b$.
\end{lemm}
\begin{proof}
  It is enough to show that $df_b = \iota_{Y_b}\Omega$, since then it follows that $\iota_{X_b}\Omega = df_b = \iota_{Y_b}\Omega$, and hence $X_b = Y_b$ by the non-degeneracy of $\Omega$.

Since $M$ is a linear manifold, each $\xi_c\in M$ corresponds to a vector $Z$ at each point $\xi_a$, which is defined by,  $Z(f) := \tfrac{d}{ds}f(\xi_a + s\xi_c)|_{s=0}$. Moreover, for any vector $Z$ and any function $f$ at a point we have $df(Z) = Z(f)$. So, let $Z$ be any vector in the tangent space of $\xi_a$. Then,
\begin{align*}
  df_b(Z) & = Z(f_b) = \tfrac{d}{ds}f_b(\xi_a + s\xi_c)|_{s=0} \\
       & = \tfrac{d}{ds}\tfrac{1}{2} g(\xi_a + s\xi_c, \xi_{ba} + s\xi_{bc})|_{s=0}\\
       & = \tfrac{d}{ds}\tfrac{1}{2}\left(
             g(\xi_a,\xi_{ba}) + sg(\xi_a,\xi_{bc})
             + sg(\xi_c,\xi_{ba}) + s^2g(\xi_c,\xi_{bc}) \right)|_{s=0}\\
       & = \tfrac{1}{2}\left( g(\xi_a,\xi_{bc}) + g(\xi_c,\xi_{ba}) \right).
\end{align*}
  But from Lemma \ref{lemm:sa} we have $g(\xi_a,\xi_{bc}) = g(\xi_{ba},\xi_c)$. So, these last two terms combine and we get,
\[
  df_b(Z) = g(\xi_c,\xi_{ba}) = \Omega(\xi_c, J\xi_{ba}) = \Omega(-J\xi_{ba},\xi_c) = \iota_{Y_b}\Omega(Z).
\]
\end{proof}

The final step is to show that the mapping $\pi:a\mapsto f_a$ defines a K\"ahler representation in the sense that,
\begin{align}
\begin{split}\label{eq:isomorph}
   f_{a\circ b} & = f_a\circ f_b\\
   f_{\{a,b\}} & = \{f_a, f_b\}\\
   |f_a| & \leq |a|.
\end{split}
\end{align}
These properties may be expressed in a more convenient form using our Lemmas. The algebra on $C^\infty(M)$ is defined by $f_a\circ f_b := \tfrac{1}{2}g(X_a,X_b)$ and $\br{f_a,f_b} := \tfrac{1}{2}\Omega(X_a,X_b)$, where $X_a,X_b$ are the Hamiltonian vector fields generated by $f_a$ and $f_b$, respectively. At each point $\xi_c\in M$, Lemma \ref{lemm:schrodinger} guarantees that $X_a = Y_a = -JX_{ac}$. We thus have that at each point $\xi_c\in M$,
\begin{align}
\begin{split}\label{eq:lemm-schro-app}
  f_a\circ f_b & = g(X_a,X_b) = g(-JX_{ac},-JX_{bc}) = g(X_{ac},X_{bc})\\
  \br{f_a,f_b} & = \Omega(X_a,X_b) = \Omega(-JX_{ac},-JX_{bc}) = \Omega(X_{ac},X_{bc}).
\end{split}
\end{align}
Moreover, the definition of our representation in Equation \eqref{eq:f_a2} implies that  $f_{a\circ b}(\xi_c) =  \tfrac{1}{2}g(\xi_c,\xi_{(a\circ b)c})$ and that $f_{\{a,b\}}(\xi_c) =  \tfrac{1}{2}g(\xi_c,\xi_{\{a,b\}c})$. Combing this with the calculations of \eqref{eq:lemm-schro-app} we find that to prove we have a K\"ahler representation only requires showing that,
\begin{align*}
  g(\xi_c,\xi_{(a\circ b)c}) & = g(\xi_{ac},\xi_{bc})\\
   g(\xi_c,\xi_{\{a,b\}c}) & = \Omega(\xi_{ac},\xi_{bc}).
\end{align*}
The proof of these statements are established in our third Lemma.

\begin{lemm}
   The mapping $\pi:b\mapsto f_b$ defined by $f_b(\xi_c)=f_b(J\xi_c) = \tfrac{1}{2}g(\xi_c,\xi_{bc})$ is a K\"ahler representation of $\Alg$, i.e. $f_{a\circ b}  = f_a\circ f_b$ and $f_{\{a,b\}} = \{f_a, f_b\}$.
\end{lemm}

\begin{proof}
Thanks to Lemma \ref{lemm:schrodinger}, it is enough to establish,
\begin{enumerate}
    \item[(a)] $g(\xi_c,\xi_{(a\circ b)c}) = g(\xi_{ac},\xi_{bc})$
    \item[(b)] $g(\xi_c,\xi_{\{a,b\}c}) = \Omega(\xi_{ac},\xi_{bc})$.
\end{enumerate}
We begin by expanding the left-hand-side of (a) as,
\begin{align*}
  g(\xi_c,\xi_{(a\circ b)c})
       & = g(\xi_c,\xi_{(a\circ b)\circ c}) + g(\xi_b,J\xi_{\br{a\circ b,c}})\\ 
       & = g(\xi_c,\xi_{(a\circ b)\circ c}) - \Omega(\xi_c,\xi_{\br{a\circ b,c}})\\
       & = \varphi( c\circ ((a\circ b)\circ c) - \br{c, \br{a\circ b,c}} ).
\end{align*}
Applying a JLB identity (Appendix JLB-\ref{jlb:3}) we now get,
\begin{align*}
  & = \varphi\big((a\circ c)\circ(b\circ c) + \br{\br{a,c},b\circ c}\big) + \varphi\big(\br{a,c}\circ\br{b,c} -\br{a\circ c,\br{b,c}}\big)\\
  & = g(\xi_{a\circ c},\xi_{b\circ c}) + \Omega(\xi_{\br{a,c}},\xi_{b\circ c})
      - \Omega(\xi_{a\circ c},\xi_{\br{b,c}}) + g(\xi_{\br{a,c}},\xi_{\br{b,c}})\\
  & = g(\xi_{a\circ c},\xi_{b\circ c}) + g(J\xi_{\br{a,c}},\xi_{b\circ c})
     + g(\xi_{a\circ c},J\xi_{\br{b,c}}) + g(J\xi_{\br{a,c}},J\xi_{\br{b,c}})\\
  & = g(\xi_{ac},\xi_{b\circ c}) + g(\xi_{ac},J\xi_{\br{b,c}})\\
  & = g(\xi_{ac}, \xi_{bc}).
\end{align*}
We next expand the left-hand-side of (b) to get,
\begin{align*}
  g(\xi_c,\xi_{\br{a,c}c})
      & = g(\xi_c,\xi_{\br{a,b}\circ c}) + g(\xi_c, J\xi_{\br{\br{a,b},c}})\\
      & = g(\xi_c,\xi_{\br{a,b}\circ c}) - \Omega(\xi_c, \xi_{\br{\br{a,b},c}})\\
      & = \varphi\big( c\circ(\br{a,b}\circ c) - \br{c, \br{\br{a,b},c}} \big).
\end{align*}
Applying another identity (Appendix JLB-\ref{jlb:4}) this becomes,
\begin{align*}
  & = \varphi\left(\br{a\circ c,b\circ c} - \br{a,c}\circ(b\circ c) + (a\circ c)\circ\br{b,c} + \br{\br{a,c},\br{b,c}}\right)\\
  & = \Omega(\xi_{a\circ c},\xi_{b\circ c}) - g(\xi_{\br{a,c}},\xi_{b\circ c})
      + g(\xi_{a\circ c},\xi_{\br{b,c}}) + \Omega(\xi_{\br{a,c}},\xi_{\br{b,c}})\\
  & =  \Omega(\xi_{a\circ c},\xi_{b\circ c}) + \Omega(J\xi_{\br{a,c}},\xi_{b\circ c})
      + \Omega(\xi_{a\circ c},J\xi_{\br{b,c}}) + \Omega(J\xi_{\br{a,c}},J\xi_{\br{b,c}})\\
  & =  \Omega(\xi_{ac},\xi_{b\circ c}) + \Omega(\xi_{ac},J\xi_{\br{b,c}})\\
  & =  \Omega(\xi_{ac},\xi_{bc}).
\end{align*}
\end{proof}

\subsection{Establishing uniqueness}

Given a state $\varphi$ on a $C^*$ algebra $\Al$, the original GNS construction produces a triple $(\HH,\pi,\psi)$ consisting of a Hilbert space $\HH$, a representation $\pi$, and a unit vector $\psi$ satisfying (i) $\varphi(A) = (\psi, \pi(A)\psi)$ for all $A\in\Al$, and (ii) the `cyclicity' condition that $\pi(\Al)\psi$ is dense in $\mathcal{H}$. One then finds that the representation is unique in the sense that any other such triple $(\HH^\prime,\pi^\prime,\psi^\prime)$ satisfying conditions (i) and (ii) is unitarily equivalent to the original. 

The K\"ahler representation constructed above displays a parallel sense of uniqueness. Given a state $\varphi$ over a JLB algebra $\Alg$, we can construct a triple $(K, \pi , \nu)$, where the K\"ahler manifold $K$ and the representation $\pi$ are defined as above, and $\nu$ is the point defined by $\nu:=\xi_{(\mathbf{1},\mathbf{1})}$, with $(\mathbf{1},\mathbf{1})\in\Alg\times\Alg$ and $\mathbf{1}$ the JLB unit. This triple will be seen to satisfy the conditions,
\begin{enumerate}
\item $f_a (\nu) = \varphi(a)$ for all $a\in\Alg$, and\label{cond:1}
\item (K\"ahler Cyclicity) Any two points in $M$ are connected by a piecewise smooth Hamiltonian curve generated by the Hamiltonian vector fields associated with $\{ f_a \}$.\label{cond:2}
\end{enumerate}
The first condition follows from the definition of $f_a$ in Equation \eqref{eq:f_a}: since in general we have that $f_a(\xi_b + J\xi_b) = 2f_a(\xi_b)$, it follows in particular that,
\begin{align*}
  f_a(\nu) & = f_a(\xi_{\mathbf{1}}+J\xi_{\mathbf{1}}) = 2f_a(\xi_{\mathbf{1}})\\
           & = g(\xi_{\mathbf{1}},\xi_{a\circ \mathbf{1}}) - \Omega(\xi_{\mathbf{1}},\xi_{a\circ \mathbf{1}}) = g(\xi_{\mathbf{1}},\xi_a) = \varphi(\mathbf{1}\circ a) = \varphi(a).
\end{align*}
The second condition follows because (as noted at the end of Section \ref{subs:construction-of-the-kahler-manifold}), at a fixed point $\xi \in M$, the tangent space $T_\xi M$ is spanned by a basis of Hamiltonian vector fields of the form $\{ X_a \equiv X_{f_a} | a \in \mathcal{A} \}$. This property is equivalent to the condition that we have called K\"ahler cyclicity, viz. that any two points in $M$ can be connected by a piecewise smooth Hamiltonian curve\footnote{See Proposition 2.3.7 of \cite{landsman1998mathematical}.}. 

To establish uniqueness, we now argue that any other such triple $(K^\prime,\pi^\prime,\nu^\prime)$ satisfying conditions \eqref{cond:1} and \eqref{cond:2} is related to $(K,\pi ,\nu)$ by a K\"ahler isomorphism $U: M \rightarrow M'$, where $M$ and $M'$ are the manifolds associated with $K$ and $K'$, $\pi':a\mapsto f'_a$, and $U^*f'_a = f_a$.

To identify $U$, first note that K\"ahler Cyclicity allows us to write an arbitrary point in $M$ as $\phi (\nu)$, where $\phi$ is the flow along a piecewise smooth Hamiltonian curve connected to $\nu$. Moreover, each smooth Hamiltonian curve is generated by some function $f_a$. So, using the functions $f'_a$ on $M'$ corresponding to the same elements of the JLB algebra, we can define a piecewise smooth Hamiltonian curve from $\nu$ to some point $\phi'(\nu')$. Thus, $U:\phi(\nu)\mapsto\phi' (\nu')$ defines a bijection from $M$ to $M'$ with the property that $U:\nu\mapsto\nu'$. 

We next show that $U$ is a K\"ahler isomorphism, in that it preserves the metric and symplectic form. Let $g'$ be the metric on $M'$, and let us write $U_*$ and $U^*$ for the pushforward and pullback of $U$, respectively. From the fact that $U(\nu)=\nu'$ we have that,
\[
  (U_*f_a)(\nu') = f_a(U^{-1}\nu') = f_a(\nu) = \varphi(a) = f'_a(\nu'),
\]
where we have applied condition \eqref{cond:1} in the last two equalities. That is, at least at the point $\nu$, we have that $U$ takes $f_a$ to $f'_a$. This implies that $(U_*X_{f_a})|_{\nu'} = X_{f'_a}|_{\nu'}$, which in shorthand we will write as $U_*X_a|_{\nu'} = X'_a|_{\nu'}$. This implies that,
\begin{equation}\label{eq:isometry1}
  (U^*g')(X_a,X_b)|_\nu = g'(U_*X_a,U_*X_b)|_{U(\nu)=\nu'} = g'(X'_a,X'_b)|_{\nu'}.
\end{equation}
But using the definition of the algebra of functions on a K\"ahler manifold and the fact that $\pi$ and $\pi'$ are representations discussed in Section \ref{sec:geometric-representation-theory}, we also have,
\begin{align}\label{eq:isometry2}
  \begin{split}
  g'(X'_a,X'_b)|_{\nu'} & = f'_a\circ f'_b(\nu') = f'_{a\circ b} =  \varphi(a\circ b)\\
        & = f_{a\circ b}(\nu) = f_a\circ f_b(\nu) = g(X_a,X_b)|_{\nu}.
  \end{split}
\end{align}
Combining Equations \eqref{eq:isometry1} and \eqref{eq:isometry2} we thus find, at least at the point $\nu$, that $U^*g'=g$. But since $\phi$ and $\phi'$ are (piecewise smooth) isometries, we immediately see that $U^*g'=g$ for all vectors of the form $X_{f_a}$, and since these vectors span the tangent space this means that $U$ is an isometry. (The argument that $U$ preserves the symplectic is exactly analogous to this.)

Finally, since we can use the representation property to write $f'_a (x)$ as $f'_{a \circ 1} (x) = f'_{a} \circ f'_{1} (x) = g' (X_{f'_a} , X_{f'_1})|_x$, the fact that $U$ is an isometry shows that $U^* f'_a = f_a$, which completes the proof that the representations $\pi$ and $\pi'$ are the same up to K\"ahler isomorphism.

\section*{Appendix: JLB Identities}

\begin{jlb}\label{jlb:1}
If $\varphi:\Alg\rightarrow\RR$ is positive ($\varphi(a\circ a)\geq 0$) and linear then it satisfies the following Cauchy-Schwarz inequalities.
\begin{enumerate}
\item[(a)] $\varphi(a\circ b)\varphi(a \circ b) \leq \varphi(a\circ a)\varphi(b\circ b)$
\item[(b)] $\varphi(\{a,b\})\varphi(\{a,b\}) \leq \varphi(a\circ a)\varphi(b\circ b)$
\end{enumerate}
\end{jlb}

\begin{proof}
  Since $\varphi$ is positive and linear, we have for all $x\in\RR$ and any $a,b\in\Alg$ that,
\[
0 \leq \varphi\big((a+bx)\circ(a+bx)\big)
   = \varphi(b\circ b)x^2 + 2\varphi(a\circ b)x + \varphi(a\circ a).
\]
This is a non-negative quadratic in $x$, and thus it has a discriminant $(2\varphi(a\circ b))^2 - 4\varphi(a\circ a)\varphi(b\circ b) \leq 0$, which implies (a). Part (b) is proved similarly; see e.g. \cite{FalcetoEtAl2013_JLB}.
\end{proof}

\begin{jlb}\label{jlb:2}
$(a\circ b)\circ(a\circ b) = \br{a,b}\circ\br{a,b} + a\circ((a\circ b)\circ b) + a\circ\br{\br{a,b},b}$ for all $a,b\in\Alg$.
\end{jlb}

\begin{proof}
  By the associator identity, $(a\circ b)\circ (a\circ b) - a\circ(b\circ (a\circ b)) = \br{\br{a, a\circ b},b}$. But the Leibniz rule allows us to write this last term,
\[
  \br{\br{a, a\circ b},b} = \br{a\circ\br{a,b},b}
      = a\circ\br{\br{a,b},b} + \br{a,b}\circ\br{a,b},
\]
which proves the claim.
\end{proof}

\begin{jlb}\label{jlb:3}
$c\circ ((a\circ b)\circ c) - \br{c, \br{a\circ b,c}} = 
   (a\circ c)\circ(b\circ c) + 
   \br{\br{a,c},b\circ c} +  
   \br{a,c}\circ\br{b,c} -
   \br{a\circ c,\br{b,c}}$ for all $a,b,c\in\Alg$.
\end{jlb}

\begin{proof}
With two applications of the associator identity to the first term we express it as,
\begin{align*}
   & = c\circ\br{\br{a,c},b} + c\circ(a\circ(b\circ c)) \\
       & = c\circ\br{\br{a,c},b} + (c\circ a)\circ(b\circ c) - \br{\br{c,b\circ c},a}\\
       & = c\circ\br{\br{a,c},b} + (c\circ a)\circ(b\circ c) - \br{c,b}\circ\br{c,a} - \br{\br{c,b},a}\circ c,
\end{align*}
where the last equation substitutes $\br{c,b\circ c} = c\circ\br{c,b} + b\circ\br{c,c} = c\circ\br{c,b}$ in the third term and then applies the Leibniz rule. Turning now to the second term, we find with two more applications of the Leibniz rule:
\begin{align*}
  & = \br{c,a\circ\br{b,c}} + \br{c,\br{a,c}\circ b}\\
      & = a\circ\br{c,\br{b,c}} + \br{c,a}\circ\br{b,c} + \br{a,c}\circ\br{c,b} + \br{c,\br{a,c}}\circ b\\
      & = - 2\br{a,c}\circ\br{b,c} - b\circ\br{\br{a,c},c} + a\circ\br{c,\br{b,c}}.
\end{align*}
We now combine by subtracting, 
\begin{align*}
& c\circ ((a\circ b)\circ c) - \br{c, \br{a\circ b,c}} = \\
   & \hspace{2pc} (a\circ c)\circ(b\circ c) + \br{a,c}\circ\br{b,c}\\
     & \hspace{5pc} + \underbrace{ b\circ\br{\br{a,c},c} + c\circ\br{\br{a,c},b} }_{\br{\br{a,c},b\circ c}}\\
     & \hspace{10pc} \underbrace{- a\circ\br{c,\br{b,c}} - \br{a,\br{b,c}}\circ c }_{-\br{a\circ c,\br{b,c}}}
\end{align*}
\end{proof}

\begin{jlb}\label{jlb:4}
$c\circ(\br{a,b}\circ c) - \br{c, \br{\br{a,b},c}} = \br{a\circ c,b\circ c} - \br{a,c}\circ(b\circ c) + (a\circ c)\circ\br{b,c} + \br{\br{a,c},\br{b,c}}$ for all $a,b,c\in\Alg$.
\end{jlb}

\begin{proof}
The second term can be rewritten using the Jacobi identity,
\[
 - \br{c, \br{a\circ b,c}} = \br{c,\br{\br{b,c},a}} + \br{c,\br{\br{c,a},b}}.
\]
The first term $c\circ(\br{a,b}\circ c)$ can be rewritten using the Leibniz rule in the form $\br{a,b}\circ c = \br{a,b\circ c} - b\circ\br{a,c}$, which gives us,
\begin{align*}
      & c\circ\br{a,b\circ c} - c\circ(b\circ\br{a,c})\\
      & = \br{c\circ a, b\circ c} - a\circ\br{c,b\circ c} - c\circ(b\circ\br{a,c})\\ 
      & = \br{c\circ a, b\circ c} - a\circ(\br{c,b}\circ c) - (c\circ b)\circ\br{a,c} + \br{\br{c,\br{a,c}},b},
\end{align*}
where the last equality applies the Leibniz rule to the second expression and the associator identity to the third. We thus find that these two terms combined give,
\begin{align*}
    &\br{c\circ a, b\circ c} - (c\circ b)\circ\br{a,c} + \underbrace{\br{\br{a,\br{b,c}},c} + a\circ(c\circ\br{b,c})}_{(a\circ c)\circ\br{b,c}} \\
    & \hspace{3pc}+ \underbrace{\br{\br{c,\br{a,c}},b} + \br{\br{\br{a,c},b},c}}_{ -\br{\br{b,c},\br{a,c}} }\\
    & = \br{a\circ c,b\circ c} - \br{a,c}\circ(b\circ c) + (a\circ c)\circ\br{b,c} + \br{\br{a,c},\br{b,c}}.
\end{align*}
\end{proof}


\end{document}